\documentclass{amsart}[12pt]
\usepackage[utf8]{inputenc}
\usepackage{amsmath, amsthm, amssymb, enumerate, amsfonts, mathrsfs, graphicx, wrapfig,float,subfig,tikz,xypic,bm,amscd,rotating}
\usepackage{hyperref}
\usepackage{tikz-cd}
\usepackage{mathrsfs}
\usepackage{slashed}
\usepackage{braket}

\usetikzlibrary{cd}

\usepackage{wasysym} 
\usetikzlibrary{snakes, shapes,arrows,patterns,decorations.pathreplacing,matrix}
\xyoption{all}

\definecolor{red}{rgb}{1,0,0}
\definecolor{green}{rgb}{0,1,0}
\definecolor{blue}{rgb}{0,0,1}

\newtheorem{theorem}{Theorem}

\newtheorem{lemma}{Lemma}

\newcommand{\scri}{\mathscr{I}}

\begin{document}
\title[Schwarzschild solution]{The conformal, complex and non-commutative structures of the Schwarzschild solution}
\author{Jonathan Holland\\ George Sparling}
\dedicatory{Dedicated to Professor James G. Holland, 1927 -- 2018}
\maketitle
\nocite{*}

\begin{abstract}

The generic null geodesic of  the Schwarzschild--Kruskal--Szekeres geometry has a natural complexification, an elliptic curve with a cusp at the singularity.  To realize that complexification as a Riemann surface without a cusp, and also to ensure conservation of energy at the singularity, requires a branched cover of the space-time over the singularity, with the geodesic being doubled as well to obtain a genus two hyperelliptic curve with an extra involution.  Furthermore, the resulting space-time obtained from this branch cover has a Hamiltonian that is null geodesically complete.  The full complex null geodesic can be realized in a natural complexification of the Kruskal--Szekeres metric.
\end{abstract}

\tableofcontents

\section{Introduction: conformal invariance of spacetime}
The eponymous space-time \cite{minkowski1909raum},  first presented by Hermann Minkowski in his great address of the 21st September 1908 at the 80th Assembly of German Natural Scientists and Physicians,  is an affine four-space equipped with a constant Lorentzian metric.   Its non-trivial geodesics, null or non-null,  are just the straight lines of the affine structure.  As described and illustrated by Minkowski, a vital feature  of this  geometry is that it is anisotropic, in that at any given point the null cone, which is ruled by the null geodesics through that point,  separates points to the past, points to the future and points which are inaccessible to any signal from or to the given point. 

Minkowski's theory was generalized in 1913 by Albert Einstein and Marcel Grossmann \cite{einstein1913entwurf}, \cite{einstein1914kovarianzeigenschaften},  the key insight of Einstein \cite{einstein1914bemerkungen} being that that a Lorentzian metric in four dimensions, no longer necessarily constant,  should describe the gravitational  potential field, such that the trajectories of free particles are the (timelike or null) geodesics of the metric.  The theory was completed in 1915 by Einstein \cite{einstein1901feldgleichungen}, who realized that the effect of matter on the gravitational field could be fully described, by equating the Einstein tensor of the metric to the energy momentum tensor of the matter, with an appropriate coupling constant.  David Hilbert \cite{hilbert1924grundlagen} showed that the resulting theory could be derived from a simple action principle, adding to the matter Lagrangian density,  the Ricci scalar of the metric times the invariant volume four-form to represent the Lagrangian density of gravity.

The Riemann curvature of a metric in dimension at least three splits naturally into two parts, one controlled by the Einstein tensor, the other not directly involving the Einstein tensor.  In three dimensions the second part vanishes identically.  In four dimensions, the Riemann tensor has twenty components, ten of these encapsulated in the Einstein tensor.  The other ten components, form a tensor, called the  Weyl tensor, which is regarded as describing the free gravitational degrees of freedom.  It can be shown that the Weyl tensor is conformally invariant: invariant under scalings of the metric, which preserve ratios of intervals, so, in particular, angles, but not the intervals themselves.  Since this realization in \cite{weyl1918reine}, \cite{schouten1921konforme}, there has been an intense focus on the question of how much of the basic physics is conformally invariant.   In particular a zero interval is conformally invariant and it emerges easily that null geodesics are conformally invariant objects. 

Minkowski  space-time can be recast in several different (but related) ways, emphasizing its (conformally flat) conformal structure.  First it can be conformally compactified with the addition of a null cone at infinity (denoted $\mathscr{I}$, or scri).   Second it embeds conformally as a part of the Einstein cylinder, the product $\mathbb{S}^1\times \mathbb{S}^3$ of a timelike Euclidean circle, $\mathbb{S}^1$, and a spacelike Euclidean three-sphere, $\mathbb{S}^3$, orthogonal to the timelike circle.  Third it may be complexified and compactified as the Klein quadric, $\mathcal{K}$ of all complex projective lines in complex projective three-space,  equivalently  the Grassmannian of all complex subspaces, of  complex dimension two,  of a complex vector space  $\mathbb{T}$ of complex dimension four.  A null geodesic of $\mathcal{K}$ is represented by a pair of subspaces $\mathbb{Z}$ and $\mathbb{W}$,  of the vector space $\mathbb{T}$, such that $\mathbb{Z}$ has dimension one, $\mathbb{W}$ has dimension three and $\mathbb{Z} \subset \mathbb{W}$.   Then the points of the null geodesic,  $(\mathbb{Z}, \mathbb{W})$ are the two-dimensional subspaces, $\mathbb{X}$, such that $\mathbb{Z} \subset \mathbb{X} \subset \mathbb{W}$.  The complex null geodesic is a Riemann sphere in  $\mathcal{K}$.    In particular, this applies to the null geodesics of real Minkowski space-time: each such may be represented as a circle on its corresponding Riemann sphere, in the complex Minkowski space-time.   In the Einstein cylinder each null geodesic traverses the circle of the $\mathbb{S}^1$ factor and simultaneously a great circle on the $\mathbb{S}^3$ factor, each at a common uniform speed, such that each circle is completed in the same affine parameter time interval.  Unwrapped in the time direction to the product $\mathbb{R} \times \mathbb{S}^3$, the null geodesic winds infinitely often around a great circle of the $\mathbb{S}^3$, as it advances in the time direction.

In each of these versions, the space of null geodesics is an homogeneous five-manifold (complex in the case of the Klein quadric) \cite{PenroseRindler2}.

\section{Mass: the metric of Karl Schwarzschild}
The passage from a conformally invariant theory to a metric theory, such as that of Albert Einstein entails the breaking of conformal invariance.  At the group theoretic level, this breaking is the reduction of the fifteen-dimensional non-compact conformal symmetry group $\mathbb{SU}(2, 2)$ to either the Poincar\'{e} group or the de Sitter or anti-de Sitter groups, each a ten-dimensional Lie group,  corresponding to the three basic geometries which underly quantum field theory and in particular allow the description of particles of a fixed non-zero rest mass.   It seems that this breaking is tied to quantum mechanics, in that physical clocks are typically constructed using quantum particles (either massless or massive).   So here we think of unparametrized null geodesics as classical and fully conformally invariant, their affine parametrization arising in connection with quantum mechanics, when the conformal invariance is broken.

In Einstein's general relativity, mass, regarded as a source of gravitational attraction arises at the most basic level in the famous metric unveiled on January 13th 1916 by Karl Schwarzschild.   The Schwarzschild metric is a static, spherically symmetric Lorentzian metric in four dimensions, with vanishing Einstein tensor.  Physically, it represents the gravitational field in empty space of a spherically symmetric body.  Following equation (14) of the paper of Schwarzschild, his metric  may be written out as follows:
\[ g =  r^2(g_\Sigma -g_2).  \]
Here $g_2$ is the metric of the unit two-sphere, $\mathbb{S}^2$,  in Euclidean three-space, $\mathbb{R}^3$.   Also $g_\Sigma$ is a Lorentzian "string'' metric defined on a two-dimensional space, $\Sigma$, with co-ordinates $(t, r)$, given explicitly by the formula:
\[ g_\Sigma = \frac{(r - 2m)}{r^3}\left(dt  + \frac{rdr}{r - 2m}\right)\left(dt  - \frac{rdr}{r - 2m}\right). \]
Here the (constant) mass parameter is $m$,  which is taken to be real and positive.  Also  $r > 2m$ is real and $t$ is real.

Using the standard dot product of $\mathbb{R}^3$,  the metric $g_2$ is as follows:
\[  g_2 =   (dx).(dx), \hspace{7pt} x.x = 1,\hspace{7pt}x.dx = 0, \hspace{7pt} x \in \mathbb{R}^3.\]
The metric has an apparent singularity at $r= 2m$.  This special value of $r$ defines the event horizon of the black hole, and is called the Schwarzschild radius.  Initially following the introduction of the solution, it was incorrectly believed that the Schwarzschild radius was a singular point for the geometry.  But it is merely a coordinate singularity \cite{lemaitre1933univers}, \cite{synge1950gravitational}. An observer falling into the black hole will cross the horizon as if it were a perfectly ordinary region of space-time.   However, a cosmological observer will see the matter indefinitely redshifted and not crossing the horizon.  From within the horizon, no signal may escape to infinity.  Inside the horizon, we can again use the Schwarzschild metric, which shows that there is a curvature singularity at $r = 0$.   Note that the mass parameter $m$ can be eliminated by rescaling the co-ordinates:  $r \rightarrow 2mr, \hspace{7pt}$ $t \rightarrow 2mt,\hspace{7pt} $  $x \rightarrow x$,   in which case we get:
\[  \frac{g}{4m^2 r^2} =  g_\Sigma - g_2, \]
  \[ g_\Sigma =  \frac{(r - 1)}{r^3}\left(dt  + \frac{rdr}{r - 1}\right)\left(dt -  \frac{rdr}{r - 1}\right),   \hspace{10pt}  (t, r) \in \mathbb{R}^2,  \hspace{10pt}r > 1,   \]
\[ g_2 = (dx).(dx),\hspace{7pt}  x.x =1, \hspace{7pt}  x.dx =0, \hspace{7pt} x\in \mathbb{R}^3.\]
The main objective of this work is to understand the nature of the null geodesics of Schwarzschild from a conformally invariant perspective.  Accordingly, we may eliminate the overall factor of $4mr^2$ by a conformal transformation, so it suffices to study the  metric $g$, conformal to that of Schwarzschild, given by:
\[ g = g_\Sigma - g_2.\]
Note that this metric has no free parameters: there is a \emph{\textbf{universal}} such model.  A fortiori there is no natural flat space limit.  The space-time is now a direct product $\Sigma \times \mathbb{S}^2$, where $\Sigma$ is the two-dimensional string space, equipped with the metric $g_\Sigma$ and $\mathbb{S}^2$ is the unit two-sphere in three dimensions,  equipped with metric $-g_2$,  the negative of its standard metric.

\begin{centering}
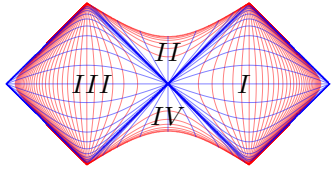
\begin{figure}
\begin{turn}{45}
\begin{tikzpicture}
\node at (0,0) {\includegraphics[scale=0.25]{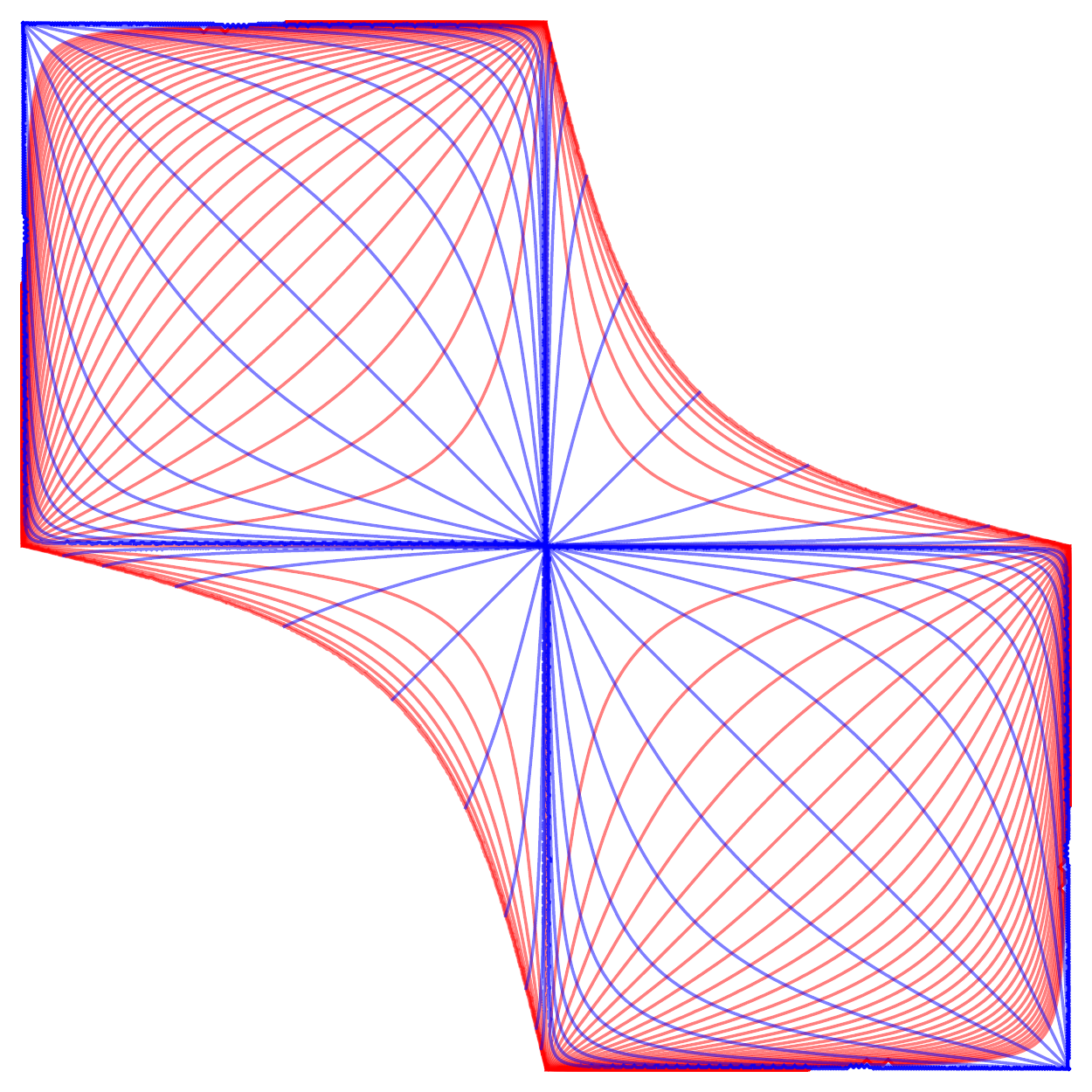}};
\node[rotate=-45] at (1/1.4,-1/1.4) {$I$};
\node[rotate=-45] at (0.3,0.3) {$II$};
\node[rotate=-45] at (-1/1.4,1/1.4) {$III$};
\node[rotate=-45] at (-0.3,-0.3) {$IV$};
\end{tikzpicture}
\end{turn}
\caption{Penrose--Carter diagram of the Kruskal--Szekeres string, with $t,r$ level curves shown, and labeled quadrants.  The ``physical'' space-time of Schwarzschild comprises quadrants $I,II$.}
\end{figure}
\end{centering}

We record some other important descriptions of the string metric, $g_\Sigma$.    First  put 
$$u = t + r  + 2^{-1} \ln\left((r- 1)^2\right),$$
so 
$$du = dt + \frac{rdr}{r - 1}$$ 
and $$v = - t + r  + 2^{-1} \ln\left((r- 1)^2\right),$$ 
so 
$$dv = - dt + \frac{rdr}{r -1}.$$
Then we have:
\begin{align*}
 g_\Sigma &=    \frac{(1 - r)}{r^3}(du)(dv) \\
& =  r^{-3}(r -1)(du)^2-2r^{-2} (du)(dr)  \\
& =  r^{-3}(r -1)(dv)^2-2r^{-2} (dv)(dr).
\end{align*}
Putting $r = s^{-1}$, we get instead:
\[ g_\Sigma =    s^2(s - 1)(du)(dv) =  2(du)(ds) + s^2(1 -s)(du)^2 = 2(du)(ds) + s^2(1 -s)(du)^2.\]
These co-ordinate systems are due to Arthur Eddington \cite{1924Natur.113..192E} and David Finkelstein \cite{1958PhRv..110..965F} and show in particular that $r = 1$ is not a singularity of the metric and that $g_\Sigma$ extends smoothly through $r = \infty$ (so through $s = 0$).   The hypersurface $s = 0$ is called conformal infinity, $\mathscr{I}$, or scri and has two pieces,  denoted $\mathscr{I}^-$ for the $(u, s)$ coordinate system (here $v$ is infinite at $s =0$) and $\mathscr{I}^+$  for the $(v, s)$ system (here $u$ is infinite at $s = 0$).   Next we put:
\[ p = e^{\frac{u}{2}} = e^{\frac{1}{2}(r + t - 1)}\sqrt{r -1}, \hspace{10pt}  q = e^{\frac{v}{2}} = e^{\frac{1}{2}(r - t - 1)}\sqrt{r -1}.\]
These are the coordinates of  Martin Kruskal \cite{kruskal1960maximal} and George Szekeres \cite{szekeres1960singularities}.   Then we have:
\begin{equation}\label{KSg} g_\Sigma =  -\frac{ 4 e^{-W(pq)} dpdq}{(1 + W(pq))^3}.\end{equation}
Here $W(z)$ is the real-valued Lambert $W$ function, defined for  $z$ real and $ez \ge -1$, implicitly by the formula:
\[   e^{W(z)} W(z) = z.\]
Note that $r = 1 + W(pq) $ and $pq^{-1} = e^t$.  This metric is real analytic in the region of the $(p, q)$ plane given by the inequality $epq > -1$.

\section{The cotangent bundle}
Passing to the cotangent bundle of the spacetime, the Hamiltonian $\mathcal{H}$ for the geodesics may be written $\mathcal{H}= H_\Sigma - H_2$, where $H_\Sigma$ is the Hamiltonian for the string metric $g_\Sigma$ and $H_2$ is that for the standard sphere metric.  Then $H_\Sigma, H_2$ and $\mathcal{H}$ pairwise Poisson commute.  Also $H_2 \ge 0$.  Along a geodesic of $g$, $\mathcal{H}$, $H_\Sigma$ and $H_2$ are constant.  For the null geodesics of the space-time, we need $\mathcal{H}$ to have the value zero, so $H_\Sigma = H_2$.  

Consequently,  a null geodesic of the spacetime may be described as a pair of geodesic curves, one in the space   $\Sigma$ and one in the sphere, each traced with a common affine parameter, each with the same non-negative energy.   

There is a special case that $H_\Sigma = H_2 = 0$:  this is the case that the geodesic on the two-sphere is just a point.  Then on the string side, we have either a point, which means that the the Schwarzschild null geodesic is a fixed point for all parameter time (a trivial case that we do not consider further), or we have a null geodesic in $\Sigma$, which does not change its angular position on the sphere, so is radial.  Turning to the generic case that $H_\Sigma = H_2 \ne 0$, we need $H_\Sigma = H_2 > 0$, so we pair a (planar) great circle on the two-sphere, traced at a uniform rate, with a time-like geodesic of the string metric.  Typically the string geodesic will reach the horizon,  $r = 1$,  in finite parameter time, or go out to infinity in $r$, so, in general,  the range of the common affine parameter  is not the whole of the real line.  Our first main result is:
\begin{theorem} There exists a real analytic extension of the cotangent bundle of $\Sigma$,  equipped with a real analytic extension of the geodesic Hamiltonian $H_\Sigma$, that is timelike and null geodesically complete. \end{theorem} 

The proof involves gluing together two Kruskal-Szekeres patches to handle the  crossings of the singularity, matched appropriately with Eddington-Finkelstein patches to handle the crossings of null conformal infinity, $\mathscr{I}$, scri.  We illustrate the proof by discussing the crossing of the singularity at $r = 0$.   We simply replace the $r$-coordinate by its square root: so we put $r =x^2$.

Then the string metric near $x = 0$ is:
\[ g_{\Sigma} = x^{-6}(x^2 - 1)\left((dt)^2  - \left(\frac{2x^3 dx}{x^2 - 1}\right)^2\right).\]    
For the contact one-form $\alpha = Tdt + Xdx$, the corresponding Hamiltonian, $H_\Sigma$,  on the cotangent bundle, is:
\[  H_{\Sigma} = \frac{1}{8(x^2 - 1)}  \left(4x^{6}T^2  -  (x^2 - 1)^2X^2\right).\]
The Hamiltonian flow gives a smooth evolution near $x = 0$ and geodesics, on reaching $x = 0$,  just cross from $x$ positive to $x$ negative or vice-versa.  Note that at $x = 0$,  the metric is singular, but the geodesic flow is not.  On each side of $x = 0$, we glue this structure to a copy of the standard Kruskal-Szekeres Hamiltonian flow, using the coordinate relations discussed in the previous section.

What this change of variables does to a null geodesic is discussed in \S\ref{ResolveSingularity}.

\section{Complex periodicity and temperature}
Following the discovery by Stephen Hawking \cite{hawking1971gravitational} and Jacob Bekenstein \cite{bekenstein1973black} that dynamical black holes radiate as black bodies, Gary Gibbons and Hawking \cite{gibbons1977action} showed that one could  account for the existence of the black hole temperature, if one demanded that the Schwarzschild solution complexify nicely, in which case it had to  be periodic in the imaginary part of the time co-ordinate $t$,   with period $8\pi m$ (here the universal required imaginary period is $4\pi$ after the coordinate rescaling given above).   Their argument may be paraphrased as follows.

Consider the string metric $g_\Sigma$, just outside the horizon,   $r = 1$. Put $r = 1 + s^2$ and $t = 2i\tau$, with $s$ and $\tau$ real and $s$ small.  To lowest order in $s$ we get:
\[ g_\Sigma =  - 4( (ds)^2 + s^2(d\tau)^2).\]
This is a regular (flat) metric,  provided $\tau$ is an ordinary polar co-ordinate, so provided that $\tau$ is considered modulo integral multiples of $2\pi$.  Correspondingly $t$ is considered modulo integral multiples of $4\pi i$.   Equivalently we may regard the co-ordinate $\displaystyle{e^{\frac{t}{2}}}$ as being more fundamental.   Finally one interpets the periodicity in imaginary time as a KMS condition for thermal equilibrium and the correct Hawking temperature  follows: in the units used here,  the temperature is $(4\pi)^{-1}$.  This argument, whilst formally convincing,  may be a little suspect, in that it appears to require a double cover branched at the horizon.  We now show that this same periodicity can be inferred by studying the Schwarzschild null geodesics.  First it emerges easily that each generic null geodesic is naturally complex.   To see this most clearly, perhaps, it is convenient to introduce the Eddington-Finkelstein co-ordinate $u$ given in terms of $u$ and $r$ by:
\[ u = t + r - 1 + 2^{-1} \ln\left((r -1)^2\right), \hspace{10pt} du = dt + \frac{rdr}{r - 1}.\]
Also put $\omega = r^{-1} - 3^{-1}$.   Then we have:
\[ g_\Sigma =   2(du)(d\omega)  - 2 Q(\omega) (du)^2,  \hspace{10pt} Q(\omega) =  \frac{\omega^3}{2} - \frac{\omega}{6} - \frac{1}{27}.\]
Note that $g_\Sigma$ is now defined and smooth,  for all $(u, \omega) \in \mathbb{R}^2$.

 Let the canonical one-form of the co-tangent bundle be $\alpha = Udu + \Omega d\omega$, with $(U, \Omega) \in \mathbb{R}^2$.     Then the Hamiltonian for the geodesics of $\Sigma$ is then:
\[ H_\Sigma = \Omega(U  + Q(\omega)\Omega).\] Hamilton's equations for the affine parameter $s$, give:
\[ \frac{du}{ds} =  \Omega, \hspace{10pt} \frac{dU}{ds} = 0, \]
\[ \frac{d\omega}{ds} =  U + 2Q(\omega) \Omega, \hspace{10pt} \Omega' =   
 -  \frac{\Omega^2}{6}   (3\omega + 1)(3\omega - 1).\] 
Then $H_\Sigma = H$, a constant along the flow and  we have, in particular:
\[ \left(\frac{d\omega}{ds}\right)^2 -   U^2 = \left(\frac{d\omega}{ds} -   U\right) \left(\frac{d\omega}{ds} +  U\right)  =  4Q(\omega) \Omega  \left(U + Q(\omega) \Omega \right)  = 4Q(\omega)H.  \]
Assuming that $H > 0$, we put $z = s\sqrt{2^{-1} H}$, giving:
\[   \left(\frac{d\omega}{dz}\right)^2 = 8Q(\omega)  + \frac{2U^2}{H} =   q(\omega), \]
\[    q(\omega) = 4\omega^3 - g_2 \omega - g_3,  \hspace{10pt} g_2 = \frac{4}{3}, \hspace{10pt} g_3 =   \frac{8H - 54U^2}{27H}.\] 
The cubic $q(w)$ has pairwise distinct roots, provided $g_2^3 - 27g_3^2 \ne 0$, so here provided that $U(8H -27U^2) \ne 0$.  We focus on the generic case that $H > 0$, $U \ne 0$ and $8H - 27U^2 \ne 0$.  
Then we have:
\begin{lemma}  The $\omega$ coordinate along a generic timelike geodesic, with rescaled affine parameter $z = s\sqrt{2^{-1} H}$ is a real curve on the (complex) elliptic curve  whose equation is in the standard  Weierstrass form:
\[ y^2 = q(\omega), \]
\[ q(\omega) = 4\omega^3 - g_2 \omega - g_3, \hspace{10pt} g_2 = \frac{4}{3}, \hspace{10pt} g_3 =   \frac{8H - 54U^2}{27H}.\] 
In terms of the parameter $z$, we  have the formulae of Weierstrass:
\[  y = \wp'(z - z_0, {g_2, g_3}), \hspace{10pt} \omega = \wp(z - z_0, {g_2, g_3}).\]
Here the prime denotes differentiation with respect to $z$.  Also $\omega$ is chosen to blow up (with a double pole) at the point $z = z_0$, for some $z_0$.
\end{lemma}

  \begin{centering}
    \begin{figure}
      \includegraphics[width=200pt,height=120pt]{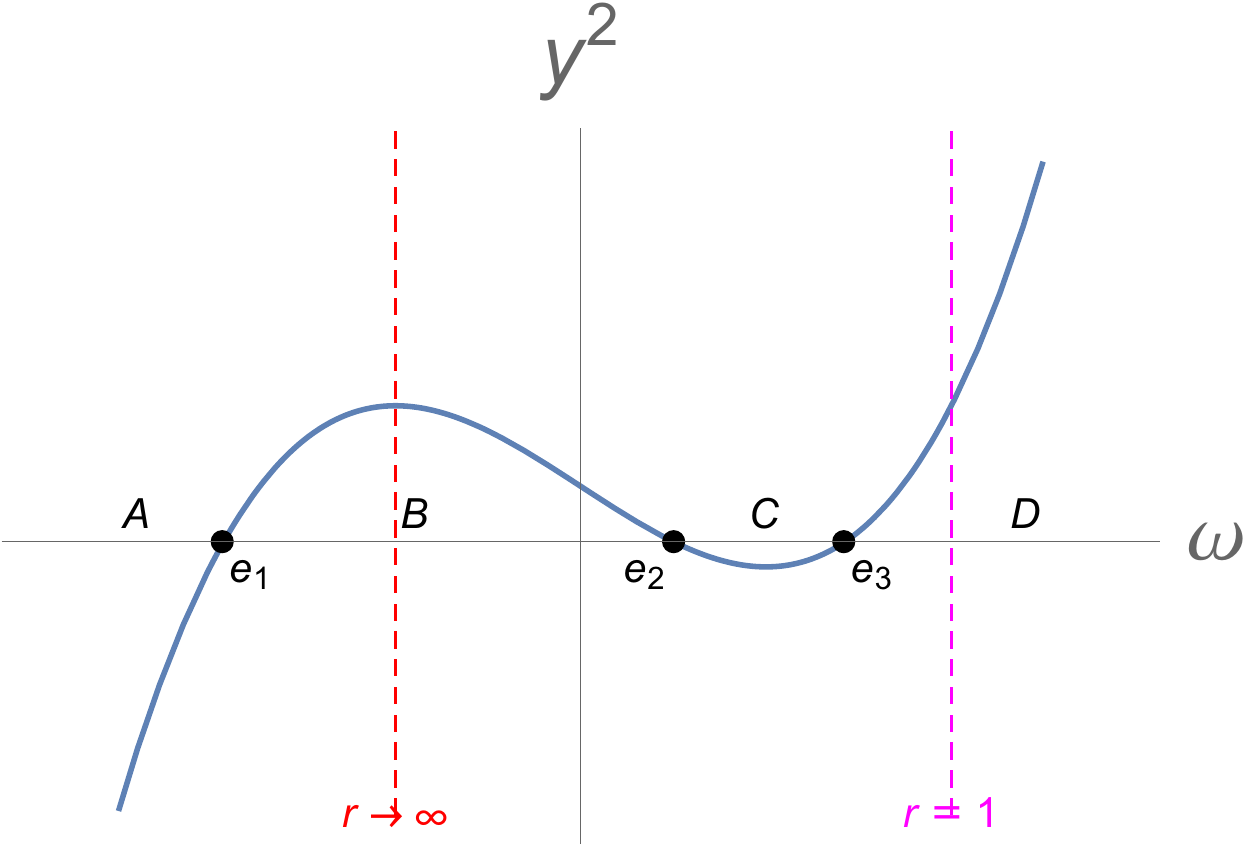}
      \caption{Plot of $y^2=q(\omega)$, in the case where $q$ has three real roots.  The singularity is at $\omega=\infty$ in these coordinates, $\omega=-1/3$ is the point at $\mathscr I$, and $\omega=2/3$ is the horizon.}
\label{cubicplot}
    \end{figure}
  \end{centering}

We turn to the coordinate $u$.  Its evolution equation is:
\[    du =  \Omega ds,  \hspace{10pt} U\Omega + Q(\omega)\Omega^2 + U\omega - H = 0.\]
Our second main theorem, the analogue of the result of Gibbons and Hawking,  in the present context,  is as follows:
\begin{theorem}\label{uTheorem} The coordinate $u$ defines a function of the Weierstrass parameter $z$, provided its values are taken in the space $\mathbb{C}/\Lambda$, where $\Lambda$ is the lattice of integer multiples of $4\pi i$.  Equivalently, the coordinate $e^{\frac{u}{2}}$ is well-defined as an holomorphic function of the complex affine parameter $z$.
\end{theorem}

The proof, conducted in \S\ref{TheCubic}, involves first showing that the meromorphic one-form $du = \Omega ds$ has one single pole and one double pole on the Weierstrass elliptic curve and then that the residue around each pole is $\pm 2$.  Thus $u$ has an imaginary period $4\pi i$.

\section{The complex metric}

\begin{figure}
\includegraphics[scale=0.5]{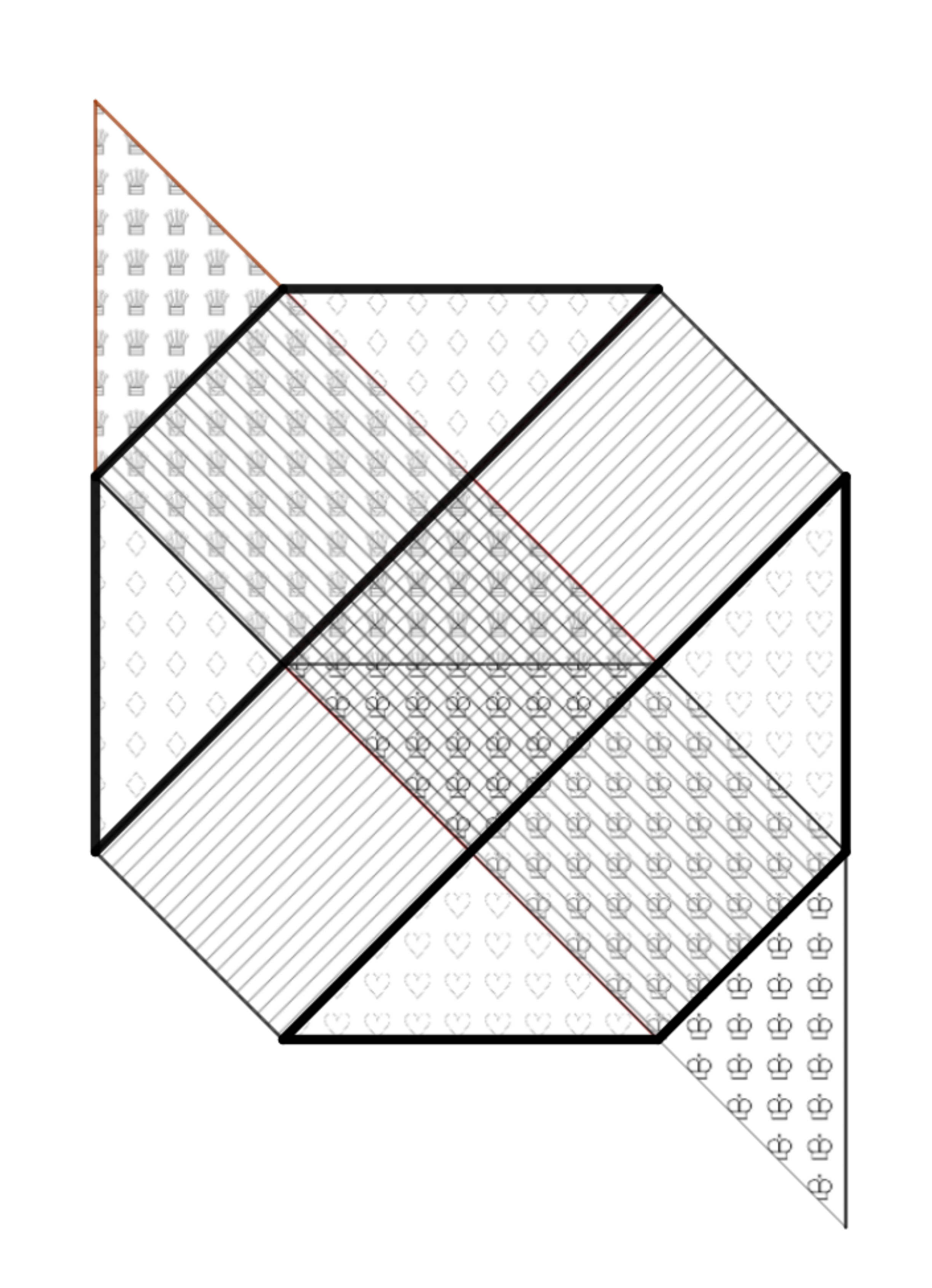}
\caption{Depiction of the extended space-time: $(u,x)$ coordinate patch (north-west hatching); $(v,x)$ coordinate patch (north-east hatching); $(u,y)$ coordinate patches (chessmen); $(v,y)$ coordinate patches (suits).  Two Kruskal diamonds are also shown, with $x>0$ and $x<0$.}
\end{figure}

Theorem \ref{uTheorem} implies that there is a natural analytic complexification of the Kruskal--Szekeres  string.  Specifically, consider at the hypersurface for $(p, q, x) \in \mathbb{C}^3$ given by the following equations:
\[ 0 = pq - (x^2 - 1)e^{x^2 - 1}, \]
\[ 0 = p dq + q dp - 2x^3e^{x^2 - 1} dx.\]
This surface is everywhere smooth since the coefficients of the above differential only  vanish when  $p = q = x = 0$, which is not a point of the surface.  In particular it is smooth at the pair of points, $p = q =0$, $x =\pm 1$.  Locally $x$ is a function of $p$ and $q$, provided $x \ne 0$.  Also $p$ is a function of $x$ and $q$, provided $q \ne 0$.  Also $q$ is a function of $x$ and $p$, provided $p \ne 0$.   When $pq \ne 0$, so when $x^2 - 1\ne 0$. we have:
\[ \frac{dp}{p} + \frac{dq}{q}  = \frac{2x^3e^{x^2 - 1} dx}{pq} = \frac{2x^3 dx}{x^2 - 1}.\]
In the complex, we cannot directly use the formula \eqref{KSg} for the metric, since it is not single-valued in the $p$ and $q$ coordinates, owing to the branching of the Lambert $W$ function.  Instead we use the formula, based on the idea that $x^2  = 1 + W(pq)$:
\[  g_\Sigma =  -\frac{ 4 dpdq}{x^6e^{x^2 -1}}  =   \frac{ 4 (1 -x^2) dpdq}{pqx^6}.\]
This metric is holomorphic and non-singular everywhere except where   $x = 0$.    We can rewrite this metric in $(p, x)$ coordinates as:
\[ g_{\Sigma} = \frac{4dp}{x^6p^2}\left((x^2 - 1)dp - 2px^3dx\right).\]
Alternatively, in $(q, x)$ coordinates, we have:
\[ g_{\Sigma} = \frac{4dq}{x^6q^2}\left((x^2 - 1)dq - 2qx^3dx\right).\]
The Hamiltonian for the canonical one-form $\alpha = Pdp + Qdq$ is:
\[ H_\Sigma = \frac{x^6pqPQ}{2(1 - x^2)}.\]
The Hamiltonian for the canonical one-form $\alpha = Rdp + Xdx$ is:
\[ H_\Sigma =  \frac{X}{8}((1 - x^2)X - 2px^3R).\]
The Hamiltonian for the canonical one-form $\alpha = Sdq + Ydx$ is:
\[ H_\Sigma =  \frac{Y}{8}((1 - x^2)Y - 2qx^3S).\]
Here we have the overlap relations:
\[ P = R + \frac{X(x^2 - 1)}{2x^3p}  =  \frac{Y(x^2 - 1)}{2x^3p}, \hspace{10pt}Q =  \frac{X(x^2 - 1)}{2x^3q}  = S + \frac{Y(x^2 - 1)}{2x^3q}, \]
\[  R =  - \frac{qS}{p}, \hspace{10pt} X = Y + \frac{2x^3 q S}{x^2 -1}, \hspace{10pt} S= - \frac{pR}{q}, \hspace{10pt}  Y = X + \frac{2x^3 p R}{x^2 -1}.\]
Next, we want to transition to where $x$ is allowed to be infinite.  Introduce the complex-valued variable $y$ which satisfies $x^2y=1$ when $y\not=0$.

Formally, we have:
\[ g_\Sigma =  \frac{ 4 dp}{x^6p^2}\left((x^2 - 1)dp - 2px^3 dx\right)  =  \frac{4dp}{p^2}\left( pdy +  (y^2 - y^3) dp\right).\]
This is globally defined and non-singular on the space of all $(p, y) \in \mathbb{C}^2$, with $p \ne 0$ and $y$ arbitrary.   Here $q$ goes to infinity, when $y$ goes to zero.  By symmetry, we also have:
\[ g_{\Sigma} =  \frac{4dq}{q^2}\left( qdy +  (y^2 - y^3) dq\right).\]
This is globally defined and non-singular on the space of all $(q, y) \in \mathbb{C}^2$, with $q \ne 0$ and $y$ arbitrary.  Here $p$ goes to infinity, when $y$ goes to zero. 

The surface is patched over the points $x=0$ and $y=0$ by projectivizing each coordinate independently, leading to the equation $x^2y=a^2b$, where $(x,a)$ and $(y,b)$ are homogeneous coordinates on a pair of copies of $\mathbb{CP}^1$.

\section{The cubic}\label{TheCubic}

\begin{lemma}
The cubic
$$q(\omega)=4\omega^3-g_2\omega - g_3=q(\omega)$$
with $g_2=4/3$, $g_3=\frac{8}{27}-\frac{2U^2}{H}$, has discriminant
$$\Delta q = 16(g_2^3-27g_3^2) =  \frac{64U^2}{H^2}(8H-27U^2)$$
\end{lemma}

There are three cases, depending on whether this discriminant $\Delta q$ is positive, negative, or zero.  In all cases, the point at scri is $\omega=-1/3$, which is a critical point of the cubic $q$.  Furthermore, the horizon is the point $\omega=2/3$, which is to the right of the second critical point of the cubic (see Figure \ref{cubicplot}).


\subsubsection*{Case 1: $\Delta q>0$}

In this case, there are three distinct real roots $e_1<e_2<e_3$.  Note that $q'(\omega)=0$ at $\omega=\pm 1/3$.  So $e_1<-1/3$, $e_3>0$, and $-1/3<e_2<1/3$.  


The three roots divide up the real $\omega$-axis into four parts, labeled $A,B,C,D$, respectively.  The regions $B$ and $D$ correspond to places where the velocity $\omega'(z)$ is real, so a portion of these curves corresponds to the real null geodesic in the original Schwarzschild geometry, being traced out for real values of the affine parameter $z$.  On the regions $A$ and $C$, the velocity is imaginary, and so a real curve is traced out for {\em imaginary} values of the affine parameter $z$.

{\centering{\begin{figure}
\hspace{1cm}\includegraphics[scale=1]{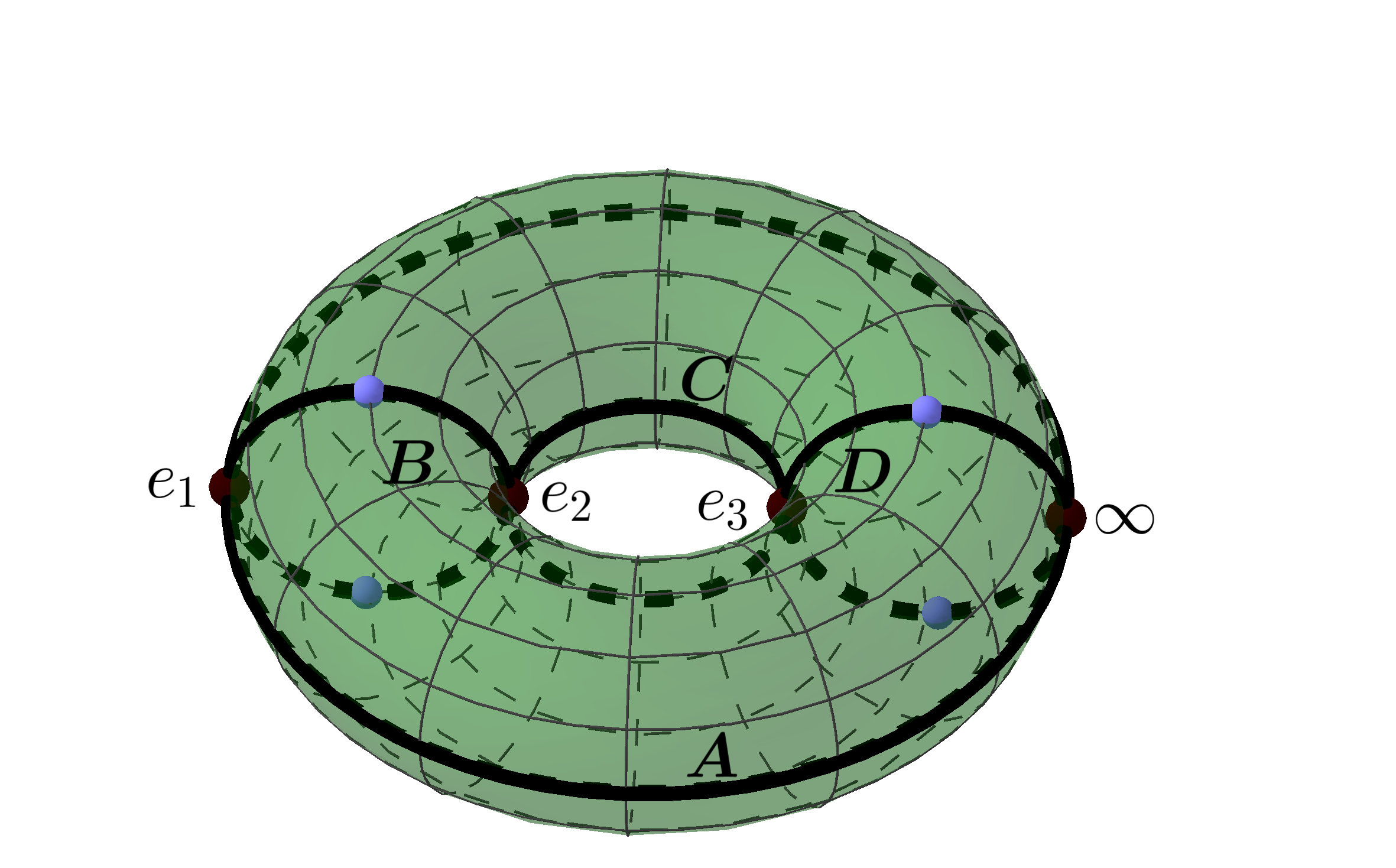}
\caption{The complexification of the null geodesic.  The roots $e_1,e_2,e_3$ and the point $\omega=\infty$ (i.e., $r=0$) are shown.  The four periods, corresponding to real values of the parameter $\omega$, are marked $A,B,C,D$, with periods $B$ and $D$ real.  The two points on $B$ where the geodesic crosses scri, and the two points on $D$ where it crosses the horizon are also marked.}
\end{figure}}}

The various bits of the null geodesic have the following interpretations in the Schwarzschild geometry.

$B$: On one portion of $B$, the geodesic begins at $\scri^-$, where $\omega=-1/3$, comes to the closest point $e_2$ to the singularity, and then heads back to $\scri^+$.  To interpret the other portion of $B$, we extend the Schwarzschild solution conformally past $\scri$ into a region of anti-Schwarzschild.  The portion of the geodesic there goes from $\scri$ to closest point $e_1$, and back to $\scri$.

$D$: the geodesic begins at the white hole singularity, reaches a furthest point $e_3$, and then falls into the black hole singularity.

$A$: an imaginary geodesic begins at the white hole singularity, reachest the furthest point $e_1$, and then falls back into the black hole singularity.

$C$: This is an imaginary geodesic that makes a cycle between $e_2$ and $e_3$.

\subsubsection*{Case 2: $\Delta q<0$}  If $\Delta q<0$, then there is a single real root $e_1$, which is negative.  There are now two parts: $A$ and $D$.  The geodesic segment $D$ contains all of the special points: $\scri$, the horizon, and the singularity.  Beginning at $\scri^-$, the geodesic crosses the horizon and falls into the black hole.  It re-emerges from the white hole, crosses the horizon and then $\scri^+$ into anti-Schwarzschild to reach a closest point $e_1$ to the singularity of anti-Schwarszchild, and then heads back out, via $\scri$, into $\scri^-$.

\subsection*{The $u$ coordinate}
The evolution of the $u$ coordinate is likewise obtained by Hamilton's equations:
$$\frac{du}{ds} = \{H,u\} = -\Omega.$$
where $\Omega$, in turn, solves the quadratic equation
\begin{equation}\label{Quadraticpomega}
Q(\omega)\Omega^2+U\Omega - H = 0.
\end{equation}
The discriminant of this quadratic polynomial in $\Omega$ is $\frac{H}{2}q(\omega)$ where $q(\omega)=4\omega^3-g_2\omega-g_3$.  This relates the two cubics $Q$ and $q$.  In addition, it implies that the branch points of the solutions to \eqref{Quadraticpomega} coincide with the branch points of the elliptic curve.

This differential equation is meromorphic in the affine parameter, and so $u$ defines a function of $z$ provided its values are taken in $\mathbb C/\Lambda$ where $\Lambda$ is the lattice generated by the residues.

To prove Theorem \ref{uTheorem}, it suffices to compute these residues:

\begin{lemma}
The meromorphic one-form $\Omega dz$ has just a single and a double pole on $E$.  At the single pole, the residue is:
$$\mathop{res}\limits_{z=z_1}\Omega\,dz = -\epsilon\sqrt{2H}$$
and at the double pole, it is:
$$\mathop{res}\limits_{z=z_2}\Omega\,dz = \epsilon\sqrt{2H},$$
where $\epsilon=\pm1$ depends on the chosen branch.
\end{lemma}
\begin{proof}
The solutions to \eqref{Quadraticpomega} are $\Omega = \frac{2H}{U \pm y\sqrt{H/2}}$ where $y$ is a solution of $y^2=q(\omega)$, so
\begin{equation*}
\Omega dz 
=\frac{-2H\,dz}{U\pm y\sqrt{H/2}}=\frac{-(8H)^{1/2}\epsilon\,dz}{\epsilon\left(\frac{2p^2_u}{H}\right)^{\!\!1/2}\!\!\! + y}
\end{equation*}

The residue at the pole $y=-\epsilon\left(\frac{2p^2_u}{H}\right)^{\!\!1/2}$ is
\begin{align*}
-(8H)^{1/2}\epsilon \lim_{z\to z_1} \frac{z-z_1}{y+\epsilon\left(\frac{2p^2_u}{H}\right)^{\!\!1/2}} &= -(8H)^{1/2}\epsilon \frac{1}{\wp''(z_1)} \\
&=\frac{-(8H)^{1/2}\epsilon}{2^{-1}q'(\omega)}.
\end{align*}
where $z_1$ is such that $\wp'(z_1)=y=-\epsilon\left(\frac{2p^2_u}{H}\right)^{\!\!1/2}$ at the pole.  

To find the value of $\omega$ at this value of $z_1$, we impose 
$$y^2=q(\omega)=\frac{2p^2_u}{H},$$
we are at a root of $Q(\omega)=0$.  These are $\omega=-1/3,2/3$.  Of these, only the latter gives a finite value of the residue, and so $z_1$ is such that $\omega=\wp(z_1)=2/3$.  We have $q'(2/3)=4$, so:
$$\mathop{res}\limits_{z=z_1}\sigma = -\epsilon\sqrt{2H}.$$

Next, for the double pole $z=z_2$, we have
$$\mathop{res}\limits_{z=z_2}\sigma = -(8H)^{1/2}\epsilon\lim_{z\to z_2}\frac{d}{dz}\left(\frac{(z-z_2)^2}{y+\epsilon\left(\frac{2p^2_u}{H}\right)^{\!\!1/2}}\right).$$
Let $c=\epsilon\left(\frac{2p^2_u}{H}\right)^{\!\!1/2}$.  Also, $y=\wp'(z)$.  Then this derivative is
$$\frac{d}{dz}\frac{(z-z_2)^2}{\wp'(z) + c} = \frac{2(z-z_2)(\wp'+c) - 2^{-1}q'(\wp(z))(z-z_2)^2 }{(\wp'(z)+c)^2}.$$
The numerator and denominator each has a zero of order four at $z=z_2$.  L'H\^opital's rule gives the value of the limit to be $4q'''(\wp(z_2))/3q''(\wp(z_2))$.  This is evaluated at the point $\omega=\wp(z_2)$ where $Q(\omega)$ has a double root, so $\wp(z_2)=-3^{-1}$.  Finally, $q'''\equiv 24$ and $q''(-3^{-1})=-8$.
\end{proof}

Recall that the affine parameter $ds$ is related to the Weierstrass parameter $dz$ by
$$ds=\sqrt{2/H}\,dz.$$
As a consequence of the theorem, the one-form 
$$du=-\Omega\,ds = -\sqrt{2/H}\Omega dz$$
has residue $2$ around each of the two poles.  Thus the coordinate $u$ has an (imaginary) period $4\pi i$.

\section{The cusp}\label{ResolveSingularity}

\begin{figure}
\centering\begin{tikzpicture}
\node {\includegraphics[scale=0.5]{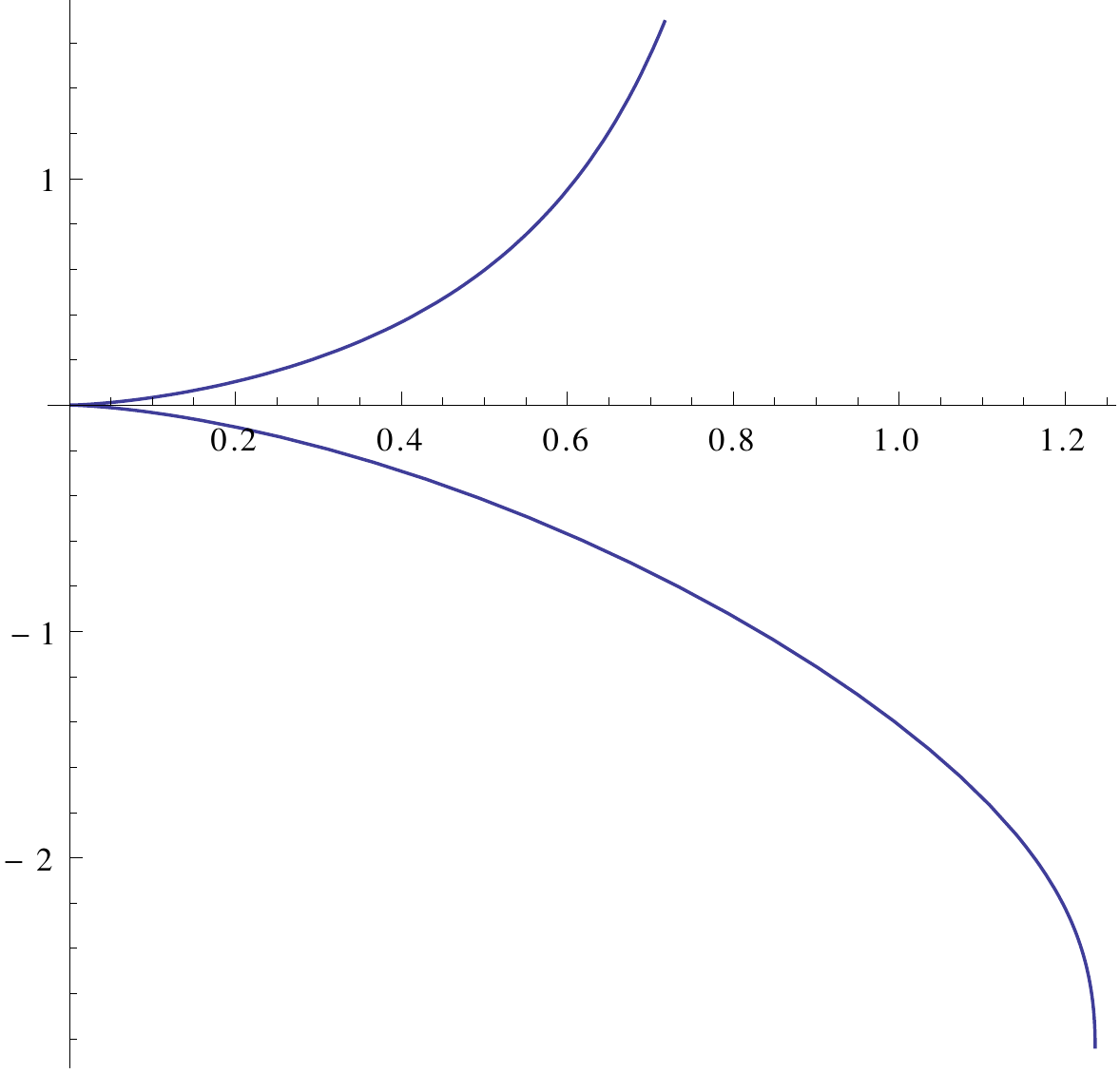}};
\node at (-2.6,3.1) {$p_v$};
\node at (3.2,0.7) {$r$};
\end{tikzpicture}
\caption{The momentum has a cuspidal singularity at $r=0$.  Shown here on a geodesic with generic initial values.}
\end{figure}

As the geodesic travels on segment $D$, beginning from the furthest point $e_3$, it first crosses the horizon and then hits the singularity.  Now, it must emerge again from the singularity.  However, at the singularity, the $(r,u)$ coordinates break down, and the geodesic has a cusp.

Indeed, in the $(r,u)$ coordinates, the Hamiltonian is
$$2H=(1-r)R^2 - 2r^2UR,$$
and $r$ satisfies $(\dot r)^2=r^4U^2+2Hr(1-r)$ which is zero at $r=0$.  On the other hand $\dot u=R$ satisfies the equation $r(1-r)(\dot u)^2 - 2r^2U\dot u -2H=0$, so that $\lim_{r\to 0^+}\dot u=-\infty$.

To resolve this cusp, put $r=x^2$.  Then the metric is
$$g_\Sigma = \frac{x^2-1}{x^6}du^2 - \frac{4}{x^4}du\,dx,$$
with Hamiltonian
$$H=\frac18(1-x^2)X^2-\frac12 x^2UX.$$


\begin{figure}
\centering\includegraphics[scale=0.5]{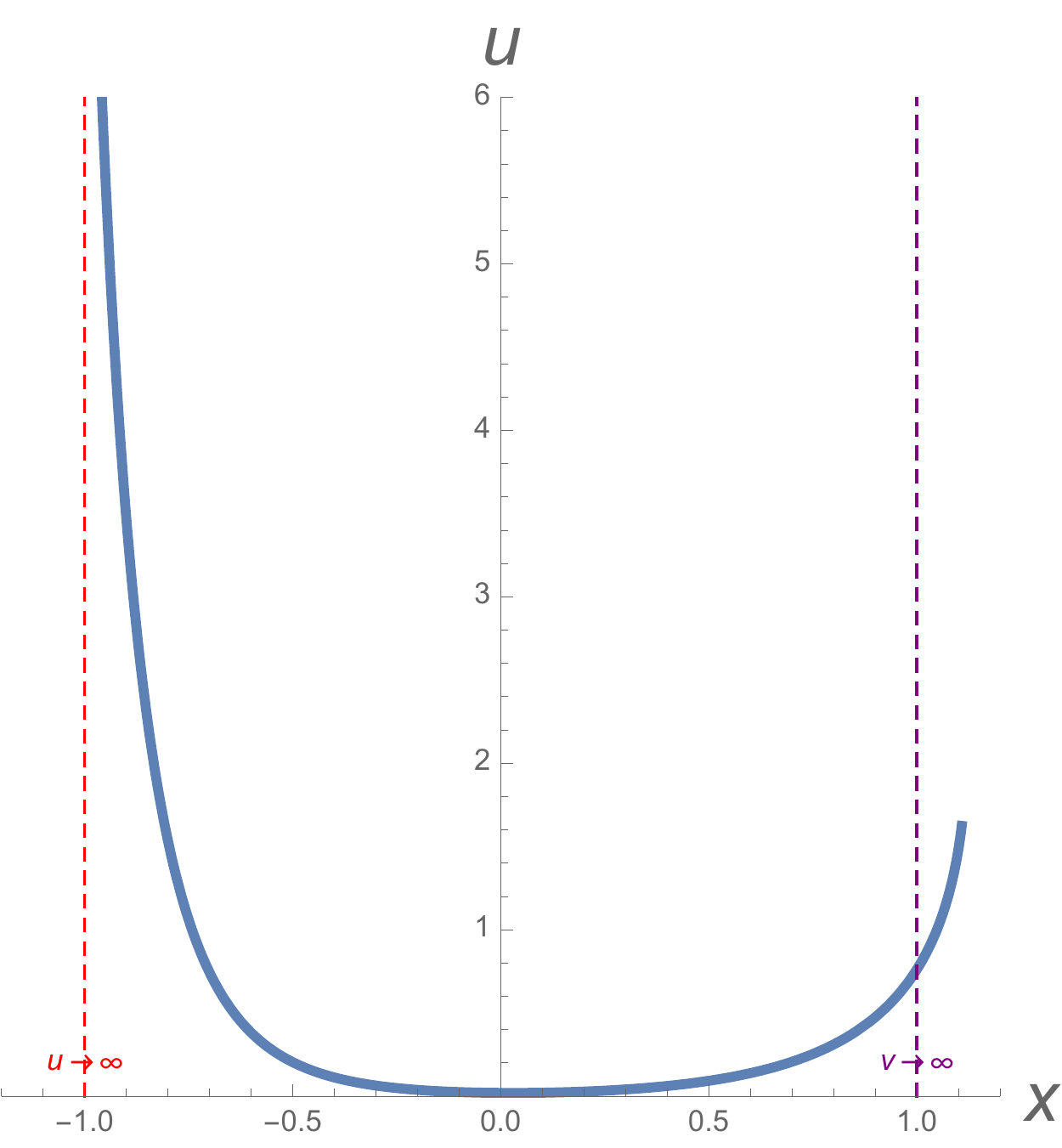}
\caption{The resolved singularity in the $(x,u)$ and $(x,v)$ coordinates.}
\end{figure}

\section{Null geodesic precession and non-commutativity}
We consider the  ``accessible'' region of the Schwarzschild solution to be the part with $t$  real and $r > 1$, so all finite points outside the horizon.   Consider the behavior of a generic null geodesic (so, in particular, non-radial) that passes through this region.  It will  eventually leave, either going out to $\mathscr{I}$, or crossing the horizon.    So for example, it may come in from $\mathscr{I}^-$, come to a point of closest approach to the horizon and then leave through $\mathscr{I}^+$.   The motion on the two-sphere $x.x= 1$ is part of a geodesic, staying in a fixed plane $n.x = 0$,  through the origin of $\mathbb{R}^3$, for $n$ some unit vector of $\mathbb{R}^3$.   If we parametrize the curve by the angle $\theta$ turned through, the angle will increase from say $\theta = \theta_1$ to $\theta = \theta_2$.   As the angular parameter $\theta$ (regarded as a real number, not considered modulo integer multiples of $2\pi$) increases further,  the trajectory will eventually return to the accessible region and trace out a new null geodesic in that region, going from say $\theta = \theta_3$, to $\theta = \theta_4$, before leaving the region again: the null geodesic has precessed.   In general this goes on for ever,   generating a countably infinite family of null geodesics in the accessible spacetime, all part of the same curve.  If tuned appropriately, it is possible to arrange that this family is finite.  As in the case of the torus $\mathbb{S}^1 \times \mathbb{S}^1$, equipped with a foliation of irrational slope, this gives rise to a non-commutative geometry in the style of Alain Connes \cite{connes2000noncommutative}. 

\begin{figure}
\definecolor{qqqqcc}{rgb}{0.,0.,0.8}
\definecolor{ffffff}{rgb}{1.,1.,1.}
\definecolor{ududff}{rgb}{0.30196078431372547,0.30196078431372547,1.}
\definecolor{ubqqys}{rgb}{0.29411764705882354,0.,0.5098039215686274}
\definecolor{ccqqqq}{rgb}{0.8,0.,0.}
\definecolor{xdxdff}{rgb}{0.49019607843137253,0.49019607843137253,1.}
\definecolor{zzttqq}{rgb}{0.6,0.2,0.}
\definecolor{uuuuuu}{rgb}{0.26666666666666666,0.26666666666666666,0.26666666666666666}
\definecolor{cqcqcq}{rgb}{0.7529411764705882,0.7529411764705882,0.7529411764705882}
\centering{
\begin{tikzpicture}[scale=0.5,line cap=round,line join=round,>=triangle 45,x=1.0cm,y=1.0cm]
\clip(-10.040679574675652,-9.129887686752454) rectangle (11.482721870214798,18.171181992607135);
\fill[line width=2.pt,color=zzttqq,fill=zzttqq,fill opacity=0.10000000149011612
] (-4.,4.) -- (4.,4.) -- (8.,0.) -- (8.,-8.) -- cycle;
\fill[line width=0.4pt,color=ffffff] (-4.,4.) -- (4.,4.) -- (8.,8.) -- (4.,12.) -- cycle;
\fill[line width=2.pt,color=ffffff,fill=ffffff,fill opacity=0.5
] (-4.,-4.) -- (4.,-4.) -- (0.,0.) -- cycle;
\fill[line width=2.pt,color=ffffff,fill=ffffff,fill opacity=0.8
] (0.,8.) -- (-4.,12.) -- (4.,12.) -- cycle;
\fill[line width=0.1pt,color=qqqqcc,pattern=dots,pattern color=qqqqcc
] (-4.,4.) -- (4.,4.) -- (-4.,12.) -- (-8.,16.) -- (-8.,8.) -- cycle;
\fill[line width=2.pt,color=ffffff,fill=ffffff,fill opacity=0.550000011920929
] (4.,4.) -- (8.,8.) -- (4.,12.) -- (0.,8.) -- cycle;
\fill[line width=2.pt,color=ffffff,fill=ffffff,fill opacity=0.800000011920929
] (0.,0.) -- (-4.,-4.) -- (-8.,0.) -- (-4.,4.) -- cycle;
\fill[line width=2.pt,color=ffffff,fill=ffffff,fill opacity=0.550000011920929
] (4.,4.) -- (8.,0.) -- (8.,8.) -- cycle;
\fill[line width=2.pt,color=zzttqq,fill opacity=0] (-8.,8.) -- (-4.,12.) -- (-8.,16.) -- cycle;
\draw [line width=2.pt] (-4.,4.)-- (4.,4.);
\draw [line width=2.pt] (-4.,-4.)-- (4.,-4.);
\draw [line width=2.pt] (4.,4.)-- (8.,0.);
\draw [line width=2.pt] (4.,-4.)-- (8.,0.);
\draw [line width=2.pt] (-8.,0.)-- (-4.,4.);
\draw [line width=2.pt] (-4.,-4.)-- (-8.,0.);
\draw[line width=2.pt,color=zzttqq, smooth,samples=100,domain=0.0:1.0] plot[parametric] function{(1.0-t)**(3.0)*5.89+3.0*(1.0-t)**(2.0)*t*4.42+3.0*(1.0-t)*t**(2.0)*4.56+t**(3.0)*5.7,(1.0-t)**(3.0)*(-2.11)+3.0*(1.0-t)**(2.0)*t*(-1.3)+3.0*(1.0-t)*t**(2.0)*1.36+t**(3.0)*2.3};
\draw[dashed,line width=2.pt,color=zzttqq, smooth,samples=100,domain=0.0:1.0] plot[parametric] function{(1.0-t)**(3.0)*5.7+3.0*(1.0-t)**(2.0)*t*6.840000000000001+3.0*(1.0-t)*t**(2.0)*6.02+t**(3.0)*5.35,(1.0-t)**(3.0)*2.3+3.0*(1.0-t)**(2.0)*t*3.2399999999999993+3.0*(1.0-t)*t**(2.0)*4.5+t**(3.0)*5.35};
\draw [line width=2.pt] (8.,-8.)-- (8.,0.);
\draw [color=ccqqqq](2.5940015094930784,-1.7649479592973094) node[scale=0.75,anchor=north west] {$\!\!\!r=1$};
\draw [color=ccqqqq](2.1178200615972718,2.6476667912038336) node[scale=0.75,anchor=north west] {$\,\,r=1$};
\draw (-0.8487960466192439,4.1) node[scale=0.75,anchor=north west] {$0^+$};
\draw (0.3083305595932076,3.314320818257963) node[scale=0.75,anchor=north west] {$D^+_{II}$};
\draw (-0.6122869063386848,-3.9553826196180206) node[scale=0.75,anchor=north west] {$0^-$};
\draw (2.1225837720181103,-0.2) node[scale=0.75,anchor=north west] {$D^+_I$};
\draw (0.5225837720181103,-2.5) node[scale=0.75,anchor=north west] {$D^+_{IV}$};
\draw (-0.8345049153567278,6.330136654931406) node[scale=0.75,anchor=north west] {$D_{II}^-$};
\draw [line width=2.pt,color=ccqqqq] (-4.,-4.)-- (4.,4.);
\draw (5.70505363574568,-0.8316399051890503) node[scale=0.75,anchor=north west] {$\mathscr I^-$};
\draw (5.546326486447078,7.155517831284138) node[scale=0.75,anchor=north west] {$\mathscr I^-$};
\draw (5.578071916306798,-4.885527506391591) node[scale=0.75,anchor=north west] {$\mathscr I^+$};
\draw (2.2670197900541964,0.6064318593400474) node[scale=0.75,anchor=north west] {$e_3^+$};
\draw (4.816181599673508,0.647704710041445) node[scale=0.75,anchor=north west] {$e_2^+$};
\draw (6.276471373220647,3.727011406434329) node[scale=0.75,anchor=north west] {$e_1^-$};
\draw[dashed,line width=2.pt,color=zzttqq, smooth,samples=100,domain=0.0:1.0] plot[parametric] function{(1.0-t)**(3.0)*5.89+3.0*(1.0-t)**(2.0)*t*7.359999999999999+3.0*(1.0-t)*t**(2.0)*7.2+t**(3.0)*6.0600000000000005,(1.0-t)**(3.0)*(-2.11)+3.0*(1.0-t)**(2.0)*t*(-2.92)+3.0*(1.0-t)*t**(2.0)*(-5.38)+t**(3.0)*(-6.0600000000000005)};
\draw [line width=2.pt,color=zzttqq] (-4.,4.)-- (4.,4.);
\draw [line width=2.pt,color=zzttqq] (4.,4.)-- (8.,0.);
\draw [line width=2.pt,color=zzttqq] (8.,0.)-- (8.,-8.);
\draw [line width=2.pt,color=zzttqq] (8.,-8.)-- (-4.,4.);
\draw[line width=2.4pt,color=ubqqys, smooth,samples=100,domain=0.0:1.0] plot[parametric] function{(1.0-t)**(3.0)*1.5319784370428084+3.0*(1.0-t)**(2.0)*t*2.58+3.0*(1.0-t)*t**(2.0)*2.52+t**(3.0)*1.5319784370428084,(1.0-t)**(3.0)*(-1.5319784370428084)+3.0*(1.0-t)**(2.0)*t*(-0.98)+3.0*(1.0-t)*t**(2.0)*0.78+t**(3.0)*1.5319784370428084};
\draw[line width=2.4pt,color=ubqqys, smooth,samples=100,domain=0.0:1.0] plot[parametric] function{(1.0-t)**(3.0)*1.5319784370428084+3.0*(1.0-t)**(2.0)*t*0.5439568740856169+3.0*(1.0-t)*t**(2.0)*(-0.9614866347955835)+t**(3.0)*(-1.0),(1.0-t)**(3.0)*1.5319784370428084+3.0*(1.0-t)**(2.0)*t*2.283956874085617+3.0*(1.0-t)*t**(2.0)*2.9651210898010394+t**(3.0)*4.0};
\draw[line width=2.4pt,color=ubqqys, smooth,samples=100,domain=0.0:1.0] plot[parametric] function{(1.0-t)**(3.0)*(-1.0)+3.0*(1.0-t)**(2.0)*t*(-1.0385133652044165)+3.0*(1.0-t)*t**(2.0)*0.5439568740856169+t**(3.0)*1.5319784370428084,(1.0-t)**(3.0)*4.0+3.0*(1.0-t)**(2.0)*t*5.034878910198961+3.0*(1.0-t)*t**(2.0)*5.716043125914383+t**(3.0)*6.468021562957191};
\draw [line width=2.pt] (-4.,12.)-- (-8.,8.);
\draw [line width=2.pt] (4.,12.)-- (8.,8.);
\draw[line width=2.4pt,color=ubqqys, smooth,samples=100,domain=0.0:1.0] plot[parametric] function{(1.0-t)**(3.0)*1.5319784370428084+3.0*(1.0-t)**(2.0)*t*2.58+3.0*(1.0-t)*t**(2.0)*2.52+t**(3.0)*1.5319784370428084-0.0,-((1.0-t)**(3.0)*(-1.5319784370428084)+3.0*(1.0-t)**(2.0)*t*(-0.98)+3.0*(1.0-t)*t**(2.0)*0.78+t**(3.0)*1.5319784370428084-8.0)};
\draw (2.3622560796333578,8.220571315252118) node[scale=0.75,anchor=north west] {$e_3^-$};
\draw (1.1622560796333578,8.820571315252118) node[scale=0.75,anchor=north west] {$D_I^-$};
\draw [line width=0.4pt,color=ffffff] (-4.,4.)-- (4.,4.);
\draw [line width=0.4pt,color=ffffff] (4.,4.)-- (8.,8.);
\draw [line width=0.4pt,color=ffffff] (8.,8.)-- (4.,12.);
\draw [line width=0.4pt,color=ffffff] (4.,12.)-- (-4.,4.);
\draw (6.227271644763723,-3.461746873125009) node[scale=0.75,anchor=north west] {$e_1^+$};
\draw (6.943125400274776,-4.050618909197182) node[scale=0.75,anchor=north west] {$B^+_{im}$};
\draw (5.514581056587358,0.9016681489192087) node[scale=0.75,anchor=north west] {$B_I^+$};
\draw[line width=2.pt,color=zzttqq, smooth,samples=100,domain=0.0:1.0] plot[parametric] function{(1.0-t)**(3.0)*5.35+3.0*(1.0-t)**(2.0)*t*4.816181599673511+3.0*(1.0-t)*t**(2.0)*5.006654178831833+t**(3.0)*5.877968988403071,(1.0-t)**(3.0)*5.35+3.0*(1.0-t)**(2.0)*t*6.996790681985539+3.0*(1.0-t)*t**(2.0)*8.901516473568767+t**(3.0)*10.12203101159693};
\draw[line width=2.4pt,color=ubqqys, smooth,samples=100,domain=0.0:1.0] plot[parametric] function{(1.0-t)**(3.0)*1.5319784370428084+3.0*(1.0-t)**(2.0)*t*0.30833055959320765+3.0*(1.0-t)*t**(2.0)*(-0.2630871778817601)+t**(3.0)*(-0.32657803760120085),(1.0-t)**(3.0)*(-1.5319784370428084)+3.0*(1.0-t)**(2.0)*t*(-2.3363656967722783)+3.0*(1.0-t)*t**(2.0)*(-3.4474557418624947)+t**(3.0)*(-4.0)};
\draw [line width=2.pt,color=ffffff] (-4.,-4.)-- (4.,-4.);
\draw [line width=2.pt,color=ffffff] (0.,0.)-- (-4.,-4.);
\draw (4.133635898270712,7.631699279179945) node[scale=0.75,anchor=north west] {$B_I^-$};
\draw (5.885998794062325,5.250792039700911) node[scale=0.75,anchor=north west] {$B^-_{im}$};
\draw [line width=2.pt,color=ffffff] (0.,8.)-- (-4.,12.);
\draw [line width=2.pt,color=ffffff] (-4.,12.)-- (4.,12.);
\draw [line width=2.pt,color=ffffff] (4.,12.)-- (0.,8.);
\draw[line width=2.4pt,color=ubqqys, smooth,samples=100,domain=0.0:1.0] plot[parametric] function{(1.0-t)**(3.0)*1.5319784370428084+3.0*(1.0-t)**(2.0)*t*0.30833055959320765+3.0*(1.0-t)*t**(2.0)*(-0.2630871778817601)+t**(3.0)*(-0.32657803760120085),(1.0-t)**(3.0)*9.53197843704281+3.0*(1.0-t)**(2.0)*t*10.336365696772278+3.0*(1.0-t)*t**(2.0)*11.447455741862495+t**(3.0)*12.0};
\draw [line width=2.pt] (-4.,12.)-- (4.,12.);
\draw [line width=2.pt,color=ccqqqq] (-4.,4.)-- (4.,-4.);
\draw [line width=2.pt] (-4.,4.)-- (-8.,8.);
\draw [line width=2.pt,color=ccqqqq] (-4.,12.)-- (0.,8.);
\draw [line width=2.pt,color=ccqqqq] (4.,12.)-- (0.,8.);
\draw [line width=2.pt,color=ccqqqq] (-4.,4.)-- (0.,8.);
\draw [line width=2.pt,color=ccqqqq] (0.,8.)-- (4.,4.);
\draw [line width=2.pt] (8.,-8.)-- (4.,-4.);
\draw [line width=2.pt] (4.,4.)-- (8.,8.);
\draw (-0.7392686257775665,12.84603803228228) node[scale=0.75,anchor=north west] {$0^-$};
\draw (5.155853907288755,8.252316745111838) node[scale=0.75,anchor=north west] {$e_2^-$};
\draw [line width=0.4pt,color=qqqqcc] (-4.,4.)-- (4.,4.);
\draw [line width=0.4pt,color=qqqqcc] (4.,4.)-- (-4.,12.);
\draw [line width=2.pt] (-4.,12.)-- (-8.,16.);
\draw [line width=2.pt] (-8.,16.)-- (-8.,8.);
\draw [line width=0.4pt,color=qqqqcc] (-8.,8.)-- (-4.,4.);
\draw [line width=2.pt,color=ffffff] (4.,4.)-- (8.,8.);
\draw [line width=2.pt,color=ffffff] (8.,8.)-- (4.,12.);
\draw [line width=2.pt,color=ffffff] (4.,12.)-- (0.,8.);
\draw [line width=2.pt,color=ffffff] (0.,8.)-- (4.,4.);
\draw [line width=2.pt,color=ffffff] (0.,0.)-- (-4.,-4.);
\draw [line width=2.pt,color=ffffff] (-4.,-4.)-- (-8.,0.);
\draw [line width=2.pt,color=ffffff] (-8.,0.)-- (-4.,4.);
\draw [line width=2.pt,color=ffffff] (-4.,4.)-- (0.,0.);
\draw [line width=2.pt,color=ffffff] (4.,4.)-- (8.,0.);
\draw [line width=2.pt] (8.,0.)-- (8.,8.);
\draw [line width=2.pt,color=ffffff] (8.,8.)-- (4.,4.);
\draw [line width=2.pt,color=zzttqq] (-4.,12.)-- (-8.,16.);
\draw [line width=2.pt,color=zzttqq] (-8.,16.)-- (-8.,8.);
\begin{scriptsize}
\draw [fill=uuuuuu] (8.,0.) circle (2.0pt);
\draw [fill=uuuuuu] (1.5319784370428084,-1.5319784370428084) circle (2.0pt);
\draw [fill=xdxdff] (5.89,-2.11) circle (2.5pt);
\draw [fill=xdxdff] (5.7,2.3) circle (2.5pt);
\draw [fill=xdxdff] (5.35,5.35) circle (2.5pt);
\draw [fill=xdxdff] (4.820999864295976,0.25891687944570363) circle (2.5pt);
\draw [fill=xdxdff] (6.265983731152718,3.3406715690370974) circle (2.5pt);
\draw [fill=xdxdff] (6.06,-6.06) circle (2.5pt);
\draw [fill=uuuuuu] (1.5319784370428084,1.5319784370428084) circle (2.0pt);
\draw [fill=ududff] (-1.,4.) circle (2.5pt);
\draw [fill=uuuuuu] (1.5319784370428084,6.468021562957191) circle (2.0pt);
\draw [fill=xdxdff] (2.2788521849582937,0.16785799977249777) circle (2.5pt);
\draw [fill=xdxdff] (2.2788521849582937,7.832142000227503) circle (2.5pt);
\draw [fill=ududff] (6.9113799704150765,-3.987128049477743) circle (2.5pt);
\draw [fill=xdxdff] (-0.32657803760120085,-4.) circle (2.5pt);
\draw [fill=xdxdff] (-0.32657803760120085,12.) circle (2.5pt);
\draw [fill=uuuuuu] (1.5319784370428084,9.53197843704281) circle (2.0pt);
\draw [fill=ududff] (5.101890468410992,7.822171858338271) circle (2.5pt);
\end{scriptsize}
\draw (5.451090196867916,2.0762490537288656) node[scale=0.75,anchor=north west] {$\mathscr I^+$};
\end{tikzpicture}
}
\caption{The geodesic journey through the extended space-time.  Roman symbols $A,B,C,D$ label the portions of the geodesic corresponding to the real axis in the Figure \ref{cubicplot}.  Subscripts denote the quadrant of the Kruskal--Szekeres extension.  A superscript $\pm$ denotes whether the coordinate $x$ is positive or negative (if real), or if $x$ is an imaginary number with positive or negative ordinate.  Roots are labeled $e_1,e_2,e_3$.}
\end{figure}
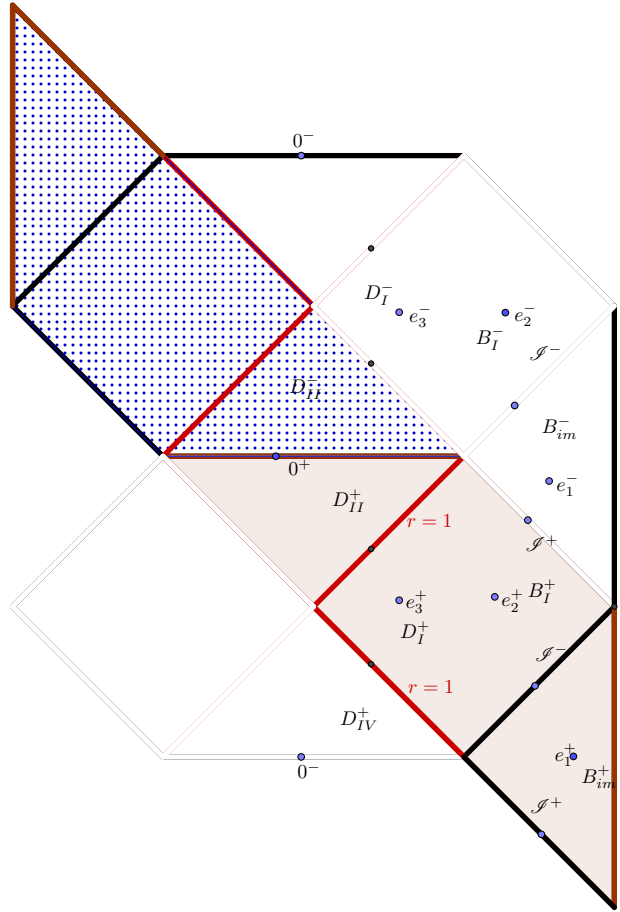

\section{Genus two: the geodesics are of Battaglini type}
At the level of  the Eddington-Finkelstein system, the geodesics complexify to elliptic curves, as described above.  However in the full system of Theorem 1 above, we have a rather different picture.  With coordinates $(u, r)$, we have:
\[ g_\Sigma =   -2r^{-2} du dr  + r^{-3}(r - 1)  du^2, \hspace{10pt} du = dt + \frac{rdr}{r - 1}.\]

Replacing $r$ by $x^2$, this gives:
\[ g_\Sigma =  - \frac{4du dx}{x^3}  +   \frac{(x^2 - 1) (du)^2}{x^6}.\]
For the contact one-form $\alpha = Udu + X dx$, we have the Hamiltonian:
\[ H_\Sigma =   \frac{X}{8}(X(1 - x^2) - 4 Ux^3) U.\]
Hamilton's equations, for the affine parameter $s$ are:
\[  \frac{du}{ds} = - \frac{Xx^3}{2},   \hspace{10pt} \frac{dU}{ds} = 0, \]
\[ \frac{dx}{ds} = \frac{1}{4}\left((1 - x^2)X - 2Ux^3\right), \hspace{10pt} \frac{dX}{ds}= \frac{xX}{4}(X + 6Ux).\]
Then $H_\Sigma = H$, a constant along the flow.  Eliminating $X$, we get the following equation for the evolution of $x$:
\[  4\left(\frac{dx}{ds}\right)^2 = U^2 x^6 - 2Hx^2 + 2H.\]
Putting $\displaystyle{y = \frac{2dx}{ds}}$, we see that the flow is a curve on the complex space with equation, for $(x, y) \in \mathbb{C}^2$:
\[ y^2 = S(x), \hspace{10pt} S(x) = U^2 x^6 - 2Hx^2 + 2H.\]
For $H > 0$, the sextic $S(x)$ has six pairwise distinct roots provided the quantity $U(8H - 27U^2)$ is non-vanishing. 

So in the generic case that $H > 0,\, U \ne 0$ and $8H - 27U^2 \ne 0$, the curve  $y^2 = S(x)$ is an hyperelliptic curve of genus two.    Now it is standard that the any genus two curve may be given by an equation of the form $y^2 = T(x)$, with $T$ a sextic polynomial, with pairwise distinct roots.  In the present case, the roots are pairwise distinct, but they come in pairs, in that if $x = a$ is a root, then so is $x = -a$ (where $0 \ne a \in \mathbb{C}$).   In other words, the curve has an extra symmetry: it is invariant under $(x, y) \rightarrow (\alpha x, \beta y)$, where $\alpha^2 = \beta^2 = 1$, in contrast to the generic curve which is only invariant under the standard hyperelliptic involution $(x, y) \rightarrow (x, \beta y)$ with $\beta^2 = 1$.   The genus two hyperelliptic curves with an extra symmetry were first analyzed in beautiful work of Giuseppe Battaglini \cite{battaglini1intorno}, \cite{battaglini1868intorno}, \cite{battaglini1866intorno}.   They may be constructed as follows:   pick two non-singular quadrics in complex projective three-space, in general position with respect to each other.   A general complex projective line then intersects the given quadrics in four points, which have an associated cross-ratio.  Battaglini requires that this cross-ratio be harmonic.    This condition is co-dimension one in the space of all lines in complex projective three-space, so gives an hypersurface in the Klein quadric, $\mathcal{K}$.    In turn it emerges that this hypersurface is itself quadratic, so the Battaglini line complex is the intersection of two non-singular quadrics in complex projective five-space, one of these the Klein quadric.     Consider the quadratic pencil defined by these quadrics, so all quadrics of the form $s\mathcal{J} + t\mathcal{K}$, where $\mathcal{J}$ is a quadric in general position with respect to the  Klein quadric $\mathcal{K}$ and $(s, t) \in \mathbb{C}^2$.   Then the associated hyperelliptic curve with the extra  symmetry is given by the equation $y^2 = \det( s\mathcal{J} + t\mathcal{K})$ and all such genus two hyperelliptic curves with an extra symmetry may be formed in this way.  So we have the Theorem:
\begin{theorem}In the extended Schwarzschild space given by Theorem 1 above, the generic null geodesics complexify to hyperelliptic curves of genus two, possessing an extra symmetry, so are associated to an harmonic line complex as discussed by Battaglini.
\end{theorem}
Such genus two curves appear in a rather different space-time context in the authors' work \cite{HollandSparlingCosmology}.

\bibliographystyle{plain}
\bibliography{schwarzschild}
\end{document}